\documentclass[a4paper,USenglish,cleveref,thm-restate]{lipics-v2021}

\usepackage{amsmath}
\usepackage{amssymb}
\usepackage{booktabs}
\usepackage{enumerate}
\usepackage{mathtools}
\usepackage[utf8]{inputenc}
\usepackage{csquotes}
\usepackage{cryptocode}
\usepackage{tikz}
\usepackage{pgfplots}
\pgfplotsset{compat=1.3}

\usepackage{thmtools}

\usepackage{hyperref}
\usepackage{cleveref}

\usepackage{multirow}

\DeclarePairedDelimiter\norm{\lVert}{\rVert}
\DeclarePairedDelimiter\abs{\lvert}{\rvert}

\newcommand{\OPT}{\operatorname{\text{\textsc{opt}}}}
\newcommand{\LB}{\operatorname{\text{\textsc{lb}\emph{(I)}}}}
\renewcommand{\pcendfor}{\textbf{endfor}}

\renewcommand{\epsilon}{\varepsilon}

\usepackage[draft]{fixme}

\usepackage{tikz}
\usetikzlibrary{decorations.pathreplacing,patterns,external}
\usepackage{acronym}
\acrodef{PTAS}[\textsc{ptas}]{polynomial time approximation scheme}
\acrodef{FPTAS}[\textsc{fptas}]{fully polynomial time approximation scheme}
\acrodef{EPTAS}[\textsc{eptas}]{efficient polynomial time approximation scheme}
\acrodef{IP}[\textsc{ip}]{integer program}
\acrodef{FFT}[\textsc{fft}]{fast Fourier transformation}
\acrodef{DFT}[\textsc{dft}]{discrete Fourier transformation}

\newcommand{\BDJR}{LRTP}

\title{Load Balancing: The Long Road from Theory to Practice} 

\author{Sebastian Berndt}{University of Lübeck, Lübeck, Germany}{s.berndt@uni-luebeck.de}{}{}
\author{Max A. Deppert}{Kiel University, Kiel, Germany}{made@informatik.uni-kiel.de}{}{Research supported by German Research Foundation (DFG) project JA 612/20-1}
\author{Klaus Jansen}{Kiel University, Kiel, Germany}{kj@informatik.uni-kiel.de}{}{Research supported by German Research Foundation (DFG) project JA 612/20-1}
\author{Lars Rohwedder}{EPFL, Lausanne, Switzerland}{lars.rohwedder@epfl.ch}{}{}
\authorrunning{S. Berndt, M.\,A. Deppert, K. Jansen, and L. Rohwedder} 

\Copyright{\ } 



\category{} 

\relatedversion{} 

\supplement{}

\nolinenumbers 

\hideLIPIcs  

\EventEditors{John Q. Open and Joan R. Access}
\EventNoEds{2}
\EventLongTitle{42nd Conference on Very Important Topics (CVIT 2016)}
\EventShortTitle{CVIT 2016}
\EventAcronym{CVIT}
\EventYear{2016}
\EventDate{December 24--27, 2016}
\EventLocation{Little Whinging, United Kingdom}
\EventLogo{}
\SeriesVolume{42}
\ArticleNo{23}
\begin{document}

\date{}

\maketitle

\begin{abstract} There is a long history of approximation schemes
  for the problem of scheduling jobs on identical machines to minimize the makespan.
  Such a scheme grants a $(1+\epsilon)$-approximation solution for every
  $\epsilon > 0$, but the running time grows exponentially in $1/\epsilon$.
  For a long time, these schemes seemed like a purely theoretical concept.
  Even solving instances for moderate values of $\epsilon$ seemed completely
  illusional.
  In an effort to bridge theory and practice, we refine recent ILP techniques to
  develop the fastest known approximation scheme for this problem.
  An implementation of this algorithm reaches values of $\epsilon$ lower than
  $2/11\approx 18.2\%$ within a reasonable timespan. This is the approximation
  guarantee of \textsf{MULTIFIT}, which, to the best of our knowledge, has the best
  proven guarantee of any non-scheme algorithm.
\end{abstract}
\vfill

\pagebreak

\section{Introduction}

Makespan minimization on identical parallel machines (often denoted by $P||C_{\max}$) asks for a distribution
of a set $J$ of $n=|J|$ jobs  to $m \leq n$ machines. 
Each job $j\in J$ has a processing time $p_j$ and the objective is to minimize the makespan, i.e., the maximum sum of processing times of jobs assigned to a single machine. 
More formally, a \emph{schedule} $\sigma\colon J\to \{1,\ldots,m\}$ assigns jobs to machines. 
The load $\ell_{\sigma,i}$ of machine $i$ in schedule $\sigma$ is defined as $\sum_{j\in \sigma^{-1}(i)} p_j$ and the \emph{makespan} $\mu(\sigma)=\max_{i} \{\ell_{\sigma,i}\}$ is the maximal load. The goal is to find a schedule $\sigma$ minimizing $\mu(\sigma)$.
This is a widely studied problem both in operations research and in combinatorial optimization and and has led
to many new algorithmic techniques.
For example, it has led to one of the earliest examples of an approximation scheme and
the use of the dual approximation technique~\cite{DBLP:journals/jacm/HochbaumS87}.
The problem is known to be strongly NP-hard and thus we cannot expect to find an exact solution in polynomial time.
Many approximation algorithms that run in polynomial time and give a non-optimal solution have been proposed for this problem. 
From a theory point of view, the strongest approximation result is a \acf{PTAS}  which
gives a $(1 + \varepsilon)$-approximation, where the precision $\varepsilon > 0$
can be chosen arbitrarily small and is given to the algorithm as input.
This goes back to a seminal work by Hochbaum and Shmoys~\cite{DBLP:journals/jacm/HochbaumS87}.
The running time of such schemes for $P||C_{\max}$ were drastically improved over time~\cite{ DBLP:conf/soda/AlonAWY97,DBLP:journals/mor/JansenKV20,DBLP:journals/ipl/Leung89} and the best known running time is $2^{O(1/\varepsilon \log^{2}(1/\varepsilon))}\log(n) + O(n)$ due to Jansen and Rohwedder~\cite{DBLP:conf/innovations/JansenR19},
which is subsequently called the JR-algorithm. The JR-algorithm is in fact an algorithm for
integer programming, but gives this running time when applied to a natural formulation of $P||C_{\max}$.
A \ac{PTAS} with a running time of $f(1/\varepsilon)\cdot n^{O(1)}$ like in the JR-algorithm
is called an \acf{EPTAS}.
It follows from the strong NP-hardness that no \acf{FPTAS}, an approximation scheme polynomial in both $n$ and $1/\epsilon$,
exists unless $\mathrm{P} = \mathrm{NP}$.

\ac{PTAS}'s are often believed to be impractical.
They tend to yield extremely high (though polynomial) running time bounds
even for moderate precisions $\epsilon$, see Marx~\cite{DBLP:journals/cj/Marx08}.
By some, the research on \ac{PTAS}'s has even been considered damaging for the large gap between theory and practice
that it creates~\cite{rosamond2006parameterized}.
Although \ac{EPTAS}'s (when \ac{FPTAS}'s are not available) are sometimes proposed as a potential solution for this
situation~\cite{DBLP:journals/cj/Marx08}, we are not aware of a practical implementation of an \ac{EPTAS}.
For example, an approximation scheme for euclidean \textsc{tsp} was implemented by Rodeker et al.,
but the algorithm was merely inspired by an \ac{EPTAS} and it does not retain the theoretical guarantee~\cite{DBLP:conf/csc/RodekerCF09}.
Although this is an interesting research direction as well,
it remains an intriguing question whether one can obtain a practically relevant \ac{EPTAS} implementation with
actual theoretical guarantees.
On the one hand, we believe that this is an important question to ask concerning
the relevance of such a major field of research.
On the other hand, such a \ac{PTAS} implementation has great advantages in itself, since it
exhibits a clean and generic design that is not specific
to any concrete precision, as well as a (theoretically) unlimited potential of the precision.

\paragraph*{Our Results.}
As a major milestone we obtain a generic \ac{PTAS} implementation that achieves in reasonable time 
a precision which beats the best known guarantee of a polynomial time non-\ac{PTAS} algorithm.
This precision to the best of our knowledge is $2/11\approx 18.2\%$, which is guaranteed by the \textsf{MULTIFIT} algorithm. 
The claim might appear vague, since the running time depends not only on $\epsilon$, but also on the instance.
We believe that it is plausible nevertheless:
The algorithm we use, which is based on the JR-algorithm,
reduces the problem to performing $O(\log(n))$ many \acfp{FFT},
where the size of the \ac{FFT} input depends only on $\epsilon$ and not the instance itself.
Hence, the running time for all instances (using the same precision) is very stable and predictable.
This is in the spirit of an \ac{EPTAS} running time.
We successfully run experiments of our implementation for a precision of $\epsilon < 2/11$
and thus make the claim that this precision is practically feasible in general.
This is also the main message of our paper. 
For completeness, we provide comparisons of the solution quality obtained empirically. While the theoretical
guarantee of the \ac{PTAS} is better,
the difference to non-\ac{PTAS} algorithms
is marginal at this state
and it is not yet evident in the experiments.
The execution of the \ac{PTAS} is computationally expensive and the considered
precision is on the edge of what is realistic for our implementation. However, we believe that further
optimization or more computational resources can lead
to also empirically superior results.
Nevertheless, the successful execution with a low precision
value forms a proof of concept for practical \ac{PTAS}'s.

Towards obtaining such an implementation we need to fine-tune the JR-algorithm significantly. In particular, it requires
non-trivial theoretical work and novel algorithmic ideas.
In fact, our variant has a slightly better dependence on the precision, namely
$2^{O(1/\epsilon\log(1/\epsilon)\log\log(1/\epsilon))}$, giving the best known running time for this problem. Our approach also greatly reduces the constants
hidden by the $O$-notation.
We first construct an \acf{IP}~---~the well-known \emph{configuration \ac{IP}}~---~that implies a $(1+\epsilon)$-approximation by rounding the processing times.
This \ac{IP} has properties that allow sophisticated algorithms to solve it efficiently.
We present several reduction steps to simplify and compress the \ac{IP} massively.
As extensions of this configuration \ac{IP} are widely used, we believe this to
be of interest in itself. 
For the makespan minimization problem, we obtain an \ac{IP} where the columns of the constraint matrix have $\ell_{\infty}$-norms bounded by $2$ and
$\ell_{1}$-norms bounded by $O(\log(1/\varepsilon))$. 
In contrast, in the classical configuration integer program used in many of the previous \ac{PTAS}'s
both of these norms are bounded by $O(1/\varepsilon)$.
This allows us to greatly reduce the size of the \ac{FFT} instances in the JR-algorithm without losing the theoretical guarantee.
For example, for $\epsilon \approx 17{.}29\%$, our reduced \ac{IP} lowers the instance sizes for \ac{FFT} from $49^{12}$ words for the configuration \ac{IP} to $5^{12}$ words.

Another important aspect in the algorithm is the rounding of the processing times. In general,
one needs to consider only $O(1/\epsilon \log(1/\epsilon))$ different rounded processing times
(to guarantee a precision of $\epsilon$).
This number has great impact on the size of the \ac{FFT} instances.
For concrete $\epsilon$ the general rounding scheme might not give the optimal number of rounded processing times.
We present a mixed integer linear program that can be used to generically optimize the rounding scheme for guaranteeing
a fixed precision $\epsilon$ (or equivalently, for a fixed number of rounded processing times).

\paragraph*{Related Work.}
The running time $f(1/\varepsilon)\cdot n^{O(1)}$ is a fixed-parameter running time, if $1/\varepsilon$ is treated as a parameter. 
In recent years, the study of practically usable parameterized algorithms has been a growing field of research.
This need for practically usable parameterized algorithms has led to the Parameterized Algorithms and Computational Experiments (PACE) challenge~\cite{DBLP:conf/iwpec/BonnetS18, DBLP:conf/iwpec/DellKTW17,  DBLP:conf/iwpec/DzulfikarFH19}. 
This challenge has brought up surprisingly fast algorithms for important problems such as \emph{treewidth}.
Note that the fastest known such algorithm due to Tamaki~\cite{DBLP:journals/jco/Tamaki19} is based on an algorithm by Bouchitt{\'{e}} and Todinca~\cite{DBLP:journals/tcs/BouchitteT02}, which was widely believed to be purely theoretic.
Our work can thus be viewed as an extension of these works to the field of approximation algorithms.

Many approximation algorithms for $P||C_{\max}$ were developed over time. 
The first such algorithm was the longest processing time first (\textsf{LPT}) algorithm by Graham, that achieved approximation ratio $4/3$~\cite{DBLP:journals/siamam/Graham69}. 
In~\cite{DBLP:journals/siamcomp/CoffmanGJ78}, Coffman et al.~presented the \textsf{MULTIFIT} algorithm that achieved a better  approximation ratio of $13/11$. 
It was later shown by Yue that this analysis is tight, i.\,e.~there are instances where \textsf{MULTIFIT} generates a solution with value $13/11 \cdot \OPT$~\cite{Yue1990OnTE}.
Kuruvilla and Palette combined the \textsf{LPT} algorithm and the \textsf{MULTIFIT} algorithm in an iterative way to obtain the Different Job and Machine Sets (\textsf{DJMS}) algorithm~\cite{DBLP:journals/ijoris/KuruvillaP15}. 


\section{Algorithm}
The general idea of our algorithm follows a typical approach for approximation
schemes.
We follow the dual approximation technique by performing a binary search on the optimal makespan.
Here it suffices to construct an algorithm that for a given value $T$ either finds a schedule
of makespan at most $(1 + \epsilon) T$ or determines that $T$ is smaller than $\OPT$. 
Then we simplify the instance such that there are no jobs of very small processing time ($\leq \varepsilon T$)
and jobs of very large processing time ($\geq (1-2\varepsilon)T$). The former is standard, whereas
the latter reduction step is novel. This already reduces the range of processing times
significantly for moderate values of $\epsilon$.
The remaining jobs are rounded via a novel rounding to $O(1/\varepsilon \log(1/\varepsilon))$
different processing times.
We can then formulate the problems as an integer program and solve it via the
algorithm of Jansen and Rohwedder~\cite{DBLP:conf/innovations/JansenR19}.
Interestingly, our new rounding scheme allows us to compress the
well-known configuration integer program quite significantly to obtain a better
running time. 

We defer some of the proofs in this section to the appendix, since they require some lengthy, but
straight-forward calculations. We write $\log = \log_2$ to denote the logarithm to base $2$.
\subsection{Rounding scheme}
It is well known that all jobs $j$ with $p_{j}\leq \varepsilon T$ can be discarded and added
greedily after solving the remaining instance. Let
$J_{\text{small}}=\{j\in J\mid p_{j}\leq \varepsilon T\}$ and
$J_{\text{large}}=J\setminus J_{\text{small}}$. 
The
feasibility of this approach follows from the following lemma:
\begin{restatable}{lemma}{LargeAndSmallJobs}
  Let $\delta_{\text{large}}$ be a schedule of $J_{\text{large}}$ with makespan $\mu(\delta_{\text{large}})$.
  Adding the jobs from $J_{\text{small}}$ greedily gives a schedule $\delta$
  with makespan
  \begin{align*}
    \mu(\delta)\leq \max\{\mu(\delta_{\text{large}}), \OPT(J)+\epsilon T\}.
  \end{align*}
  The procedure can be implemented to run in time $\mathcal{O}((|J_{\text{small}}|+m)\cdot
  \log(|J_{\text{small}}|+m))$. 
\end{restatable}
Furthermore, we can also get rid of huge jobs $J_{\text{huge}}$ with processing times at least
$(1-2\varepsilon)T$, as each such job can only be paired with at most one other job
from $J_{\text{large}}$ without violating the guess $T$.
It is easy to see that  we can pair a huge job with the largest possible
large job without losing optimality.
\begin{restatable}[informal]{lemma}{HugeJobs}
  \label{lem:huge_jobs}
  There is an optimal schedule $\delta_{\text{large}}$ of $J_{\text{large}}$
  where each huge job is paired with the largest possible large job (or not
  paired at all). 
\end{restatable}
As we now know how to place all of the jobs in $J_{\text{huge}}$ optimally, we
can ignore them and their paired jobs in the following.
After removing all of these jobs, we are left with the remaining jobs
$J_{\text{rem}}$ that we still need to schedule. 
For all $j\in J_{\text{rem}}$, we now know that we have $p_{j}\in(\varepsilon T,(1-2\varepsilon)T)$.
We will now round these remaining item sizes in order to reduce the number of
different processing times in our instance.
In order to do this, we first split the interval $(\varepsilon T,(1-2\varepsilon)T)$ into
$\log(1/\varepsilon)$ growing intervals of size $2^{i}\varepsilon T$ (starting with $i=0$).
Each of these intervals is then split into $1/\varepsilon$ smaller intervals of
the same size.

For example, for $\varepsilon=1/6$ and $T=1$, the growing intervals $(1/6,1/3]$
and $(1/3,2/3]$ are split into smaller intervals with the following boundaries.
\begin{align*}
  \frac16, \ \frac16+\frac1{36}, \ \frac16+\frac2{36}, \ \frac16+\frac3{36}, \ \frac16+\frac4{36}, \  \frac16+\frac5{36}, \\
  \frac13, \ \frac13+\frac1{18}, \ \frac13+\frac2{18}, \ \frac13+\frac3{18}, \ \frac13+\frac4{18}, \ \frac13+\frac5{18}.
\end{align*}
More formally, for $i\in \mathbb{Z}_{\geq 0}$, let $I_{i}=(2^{i}\varepsilon
T,2^{i+1}\varepsilon T]$.
In the example above, we thus have $I_{0}=(1/6,1/3]$ and $I_{1}=(1/3,2/3]$. 
We further partition an intervall $I_{i}$ into $\lceil 1/\varepsilon \rceil$ subintervals
$I_{i,k}=(b_{i,k},b_{i,k+1}]\cap I_{i}$ with $b_{i,k}=2^{i}\varepsilon
T+ k\epsilon^{2} 2^{i}T$ for $k\in \{0,\ldots, \lceil 1/\varepsilon  -1\rceil \}$.
Hence, the above exemplary boundaries are exactly the values $b_{i,k}$ for $i\in
\{0,1\}$ and $k\in \{0,\ldots,5\}$. 
The processing time of any remaining job $j\in J_{\text{rem}}$ is rounded down to the next lower
boundary.
We denote this rounded processing time of $j$ by $\tilde{p}_{j}$. 
\begin{restatable}[informal]{lemma}{rounding}
  \label{lem:rounding}    
  There are $O(1/\varepsilon \log(1/\varepsilon))$ rounded processing times
  $\tilde{p}_{j}$ and a schedule $\tilde{\sigma}$ of the rounded processing times implies a
  schedule $\sigma$ of the original processing time with $\mu(\sigma)\leq
  (1+\varepsilon)\mu(\tilde{\sigma})$.
  Furthermore, the sum of two boundaries $b_{i,k}$ and $b_{i,k'}$, where $k$ and
  $k'$ have the same parity, is equal to some boundary $b_{i+1,k''}$. 
\end{restatable}
The last property of the lemma that every two boundaries $b_{i,k}$ and
$b_{i,k'}$ of the same interval (with the same parity of $k$ and $k'$) sum up to
a boundary $b_{i+1,k''}$ in the next interval will be heavily used next.
Intuitively, this property implies that whenever a job with rounded processing
time $b_{i,k}$ and another job with rounded processing time $b_{i,k'}$ are
scheduled on the same machine, we can treat them as a single job with rounded
processing time $b_{i+1,k''}$.
This allows us to characterize the possible ways to assign rounded jobs to
machines in a more compact ways, which in turn allows us to solve the
corresponding integer program much faster.

\subsection{A new integer program}
\label{sec:ilp}
Integer programs are widely used to design approximation algorithm and
approximation schemes.
The classical result of Lenstra and Kannan~\cite{DBLP:journals/mor/Lenstra83,DBLP:journals/mor/Kannan87}  shows that an integer
program with $n$ variables can be solved in time $n^{O(n)}\cdot |I|^{O(1)}$,
where $|I|$ is the encoding length of the integer program (i.\,e.~the binary
encoding of all numbers in the objective function, the right-hand side, and the
constraint matrix).
This result was heavily used in the past to design approximation schemes.
In fact, using Lemma~\ref{lem:rounding} with this algorithm already yields an
algorithm with running time double exponential in $1/\varepsilon$.
In the past years, other parameters besides the number of variables were
studied, including the number of constraints, the largest entry in the
constraint matrix or the properties of the graph corresponding to the
constraints~(see e.\,g.~\cite{DBLP:journals/talg/EisenbrandW20,DBLP:conf/icalp/JansenLR19,DBLP:conf/icalp/KouteckyLO18}).
We will make use of the recent results that use the number of constraints
(i.\,e.~the number of rows of the constraint matrix) \emph{and} the largest
entry of the constraint matrix.

A basic concept of many algorithms for $P||C_{\max}$ is the \emph{configuration
  integer program}, which we will also use.
Roughly speaking, for each possible way $c$ to put jobs on a machine (called a
configuration), this integer program has a variable $x_{c}$ indicating
how often this configuration is used.
Then the integer program expresses that all jobs should be scheduled and that
the number of configurations used should not exceed $m$ via suitable
constraints. 

More formally, the integer program is constructed in the following way. Let
$d$ denote the number of rounded item sizes. 
We will index a vector $x\in \mathbb{Z}^{d}$ by pairs $(i,k)$, corresponding to
the values used in the rounded item sizes $b_{i,k}$ and denote its
corresponding entry by $x[i,k]$. 
A vector $c\in \mathbb{Z}_{\geq 0}^{d}$ thus describes a possible way to
schedule jobs on a machine, where the value $c[i,k]$ describes how many jobs with
rounded processing time $b_{i,k}$ are put on a machine.
We call such a vector $c$ a \emph{configuration}, if the resulting load of the
machine does not exceed~$T$, i.\,e.~$\sum_{i,k}c[i,k]\cdot b_{i,k}\leq T$.
Let $\mathcal{C}$ be the set of all configurations. 
For each $c\in \mathcal{C}$, we have a variable $x_{c}$ that describes how often
configuration $c$ is used, i.\,e.~$x_{c}$ machines are scheduled according to
$c$.
As we only have $m$ machines available, we are only allowed to use at most~$m$
configurations.
Hence $\sum_{c\in \mathcal{C}}x_{c}\leq m$.
To guarantee that all jobs are scheduled, let $n_{i,k}$ be the number
of items with rounded processing time $b_{i,k}$.
Now, summing over all chosen configurations, we want that they contain at least
$n_{i,k}$ jobs of rounded processing time $b_{i,k}$.
Hence, $\sum_{c\in \mathcal{C}}x_{c}\cdot c[i,k]\geq n_{i,k}$. 
Combining these with the natural requirement that~$x_{c}\in \mathbb{Z}_{\geq
  0}$, we obtain the following integer program 
called the \emph{configuration \ac{IP}}: 
\begin{align*}
  \sum_{c\in \mathcal{C}}x_{c} &\leq m \\
  \sum_{c\in \mathcal{C}}x_{c}\cdot c[i,k] &\geq n_{i,k} \ \  \forall (i,k) \tag{confIP} \\
  x_{c} &\in \mathbb{Z}_{\geq 0} \ \  \forall c\in \mathcal{C}
\end{align*}

As described above, the important parameters in the algorithm that we want to
use are the number of rows of the constraint matrix ($d$ in our case) and the
largest entry in the constraint matrix ($\max_{c\in \mathcal{C},
  (i,k)}\{c[i,k]\}$). 
Now, the first property of Lemma~\ref{lem:rounding} already shows that the number of
rows of the configuration \ac{IP} is bounded, i.\,e.~$d\leq O(1/\varepsilon\cdot
\log(1/\varepsilon))$.
As every boundary $b_{i,k}$ is at least $\varepsilon T$ and we aim for a maximal
load of $T$, we can easily see that the largest entry of a configuration and thus of the constraint matrix
is at most $1/\varepsilon$.
A closer look reveals that we actually have the slightly stronger bound of
$\lVert c \rVert_{1}\leq 1/\varepsilon$ for all $c\in \mathcal{C}$.
Without jumping too far ahead, the algorithm of Jansen and
Rohwedder~\cite{DBLP:conf/innovations/JansenR19} discussed in
Section~\ref{sec:rohwedder} will thus yield a running time
$2^{\mathcal{O}(1/\varepsilon\log^2(1/\varepsilon))}+O(n)$, which is slightly too high to
be usable in practice for our desired approximation guarantee of $\varepsilon <
2/11$. 
To decrease this running time, we will make use of the last property of
Lemma~\ref{lem:rounding}, which will give an improved bound of $\lVert c \rVert_{1}\leq
O(\log(1/\varepsilon))$ and thus improve the running time to
$2^{\mathcal{O}(1/\varepsilon\log(1/\varepsilon)
  \log(\log(1/\varepsilon)))}+O(n)$.
  Moreover, the hidden constants are significantly lower.
  This is a sufficient improvement for the
algorithm to run in reasonable time for $\varepsilon < 2/11$. 

To improve the bound on $\lVert c \rVert_{1}$, we will \emph{add} new columns
$\hat{\mathcal{C}}$ to the configuration \ac{IP}.
Remember that Lemma~\ref{lem:rounding} states that all for all boundaries
$b_{i,k}$ and $b_{i,k'}$ with $k\bmod 2 = k'\bmod 2$, there is
$b_{i+1,k''}=b_{i,k}+b_{i,k'}$. 
The main idea behind these new columns $\hat{\mathcal{C}}$ is that whenever we
use a job with processing time $b_{i,k}$ and a job with processing time
$b_{i,k'}$ on the same machine, we can treat this as a \emph{single job} with processing
time $b_{i+1,k''}$.
Each new column $\hat{c}(i,k,k')$ will do this exact replacement.
The final observation that we need is that in all configurations $c\in
\mathcal{C}$ with $\lVert c \rVert_{1} > 2\log(1/\varepsilon)$, we can do such a
replacement:
There are only $\log(1/\varepsilon)$ growing large intervals $i$ in our
rounding and in each interval we can choose at most two boundaries of different
parity $k\bmod 2\neq k'\bmod 2$. 
Hence, if $\lVert c \rVert_{1} > 2\log(1/\varepsilon)$, configuration $c$ uses
two jobs with processing times $b_{i,k}$ and $b_{i,k'}$ with $k\bmod 2 = k'\bmod
2$ and we can thus reduce this configuration via $\hat{c}(i,k,k')$.

By adding the columns $\hat{\mathcal{C}}$ to the configuration integer
program, we can remove all configurations $c$ except those in $\mathcal{C}_{\textrm{red}}=\{c\in \mathcal{C} \colon \vert\vert
c\vert\vert_{1}\leq 2\log(1/\epsilon)\}$.
Let us denote this \ac{IP} by
$\mathsf{IP}_{\mathcal{C}_{\textrm{red}},\hat{\mathcal{C}}}$.
Our discussion above thus implies the following lemma.

\begin{restatable}[informal]{lemma}{reducedConfiguration}
  \label{lem:reduced}
  For all solutions to the integer program
  $\mathsf{IP}_{\mathcal{C}_{\textrm{red}},\hat{\mathcal{C}}}$, we can compute
  in linear time a solution to (confIP) and vice versa. 
\end{restatable}

The advantage the system 
$\mathsf{IP}_{\mathcal{C}_{\textrm{red}},\hat{\mathcal{C}}}$ gives us is that
all columns have an $\ell_{1}$-norm bounded by $O(\log(1/\epsilon))$. 
The running time of the JR-algorithm directly depends on the discrepancy of the underlying constraint matrix.
This improved bound on the $\ell_{1}$-norm then allows us to bound this discrepancy leading to a faster running time (see Sec.~\ref{sec:rohwedder} for a more thorough discussion).
Furthermore, the $\ell_{\infty}$-norm of each column is at most $2$ (due to the columns in~$\hat{\mathcal{C}}$). 
Already for relatively large values of $\epsilon$, this reduces the number of columns significantly. 
For example, for $\epsilon=1/6$, the number of columns is reduced from $409$ down to $213$.

\subsection{Applying the JR-algorithm}
\label{sec:rohwedder}
Jansen and Rohwedder~\cite{DBLP:conf/innovations/JansenR19} described an algorithm for
integer programming and applied it to the configuration \ac{IP} for $P\abs{}C_{\max}$.
This algorithm reduces the task of solving the integer program to a small number
of \acfp{FFT}.
The size of the \ac{FFT} input depends on the number of rows of the constraint matrix as well as
its discrepancy. Using the properties of our new integer program
we are able to derive much better bounds on the discrepancy of the constraint matrix.
Intuitively, the discrepancy of a matrix $A$ measures how well the value $A\cdot
(1/2,1/2,\ldots,1/2)^{T}$ can be approximated by the term $Az^{T}$, where $z$ is
some binary vector. 

\begin{definition}[Discrepancy]
For a matrix $A \in \mathbb{R}^{m\times n}$ the \emph{discrepancy of $A$} is given as
\[\operatorname{disc}(A) = \min_{z \in \{0,1\}^n} \norm*{A\left(z-\left(\frac12,\dots,\frac12\right)^T\right)}_{\infty}.\]
Moreover, the \emph{hereditary discrepancy} of $A$ is then defined as
\[\operatorname{herdisc}(A) = \max_{I\subseteq \{1,\dots,n\}} \operatorname{disc}(A_I)\]
where $A_I$ denotes the matrix $A$ restricted to the columns $I$.
\end{definition}
We sketch the main ideas of the JR-algorithm and refer to~\cite{DBLP:conf/innovations/JansenR19} for details.
The algorithm is based on the idea of splitting the solution to an \ac{IP} $\{Ax = b, x\in\mathbb Z_{\ge 0}\}$
into two parts $x' + x'' = x$ where $Ax'$ and $Ax''$ are almost the same. Hence, computing all solutions of
the \ac{IP} with $b' \in \mathsf{H}(b/2)$ we can derive a solution with $b$. Here $\mathsf{H}(b/2)$ is an axis-parallel hypercube
with sufficiently large side length surrounding $b/2$.
The algorithm then iterates this idea.
Indeed, the running time of the algorithm greatly depends on the bound of how evenly a solution can be split,
that is, how large the hypercube needs to be.
For this, discrepancy is a natural measure.
A closer inspection of the JR-algorithm shows
that it suffices for their algorithm to take
a hypercube of side length $4\operatorname{herdisc}(A) - 1$.
Thus, the total number of elements in $\mathsf{H}(b/2^i)$ is at most
$(4\operatorname{herdisc}(A))^m$ for all $i$.
The central subprocedure in the algorithm is then to try to combine
any two solutions for $b', b''\in \mathsf{H}(b/2^i)$ to a solution for
$b' + b''\in \mathsf{H}(b/2^{i-1})$.
Instead of the naive algorithm that takes quadratic time (in the
number of elements of $\mathsf{H}(b/2^i)$) it can be implemented more efficiently
as multivariate polynomial multiplication where the input polynomials have $m$ variables and maximum degree $4\operatorname{herdisc}(A) - 1$.
This in turn can be computed efficiently using \ac{FFT}
on inputs of size $\mathcal O((4\operatorname{herdisc}(A))^m)$, see
Appendix~\ref{apx:fft} for details. 

\paragraph*{Evaluating the discrepancy.} The only remaining task now is to bound
the discrepancy of our compressed integer program $\mathsf{IP}_{\mathcal{C}_{\textrm{red}},\hat{\mathcal{C}}}$ presented in
Section~\ref{sec:ilp}. 
Since its columns have small $\ell_1$-norm,
the classical Beck-Fiala theorem allows us to give a very strong bound on its
discrepancy. 
\begin{theorem}[Beck, Fiala~\cite{BECK19811}]
	For every matrix $A\in \mathbb{R}^{m\times n}$ where the $\ell_1$-norm of each column is at most $t$ it holds that $\operatorname{herdisc}(A) < t$.
\end{theorem}
Moreover, Bednarchak and Helm~\cite{DBLP:journals/combinatorica/BednarchakH97} observed 
that for $t\ge 3$ the bound can be improved to $\operatorname{herdisc}(A) \le t - 3/2$.
This can be applied directly to our bounds on the $\ell_1$-norm of the
configurations of $\mathsf{IP}_{\mathcal{C}_{\textrm{red}},\hat{\mathcal{C}}}$.
Let $A_{\mathcal{C}_{\textrm{red}},\hat{\mathcal{C}}}$ denote the corresponding
constraint matrix. 
Since for each column $c$ of $A_{\mathcal{C}_{\textrm{red}},\hat{\mathcal{C}}}$, we
have 
$\norm{c}_1 \leq 2\log(1/\varepsilon)$, we get 
\begin{equation*}
  \operatorname{herdisc}(A_{\mathcal{C}_{\textrm{red}},\hat{\mathcal{C}}}) \le \mathcal O(\log(1/\varepsilon)) .
\end{equation*}
Moreover, the number of rows $m$
is one plus the number
of rounded processing times, that is, $O(1/\epsilon \log(1/\epsilon))$.
The algorithm needs to perform $\mathcal O(\log(n))$ many
\acp{FFT} on input of size $(\log(1/\epsilon))^{\mathcal O(1/\epsilon \log(1/\epsilon))}$.
Hence, the total running time thus becomes
\begin{equation*}
    2^{\mathcal{O}(1/\varepsilon\log(1/\varepsilon)\log\log(1/\varepsilon))}\log(n) + \mathcal{O}(n) .
\end{equation*}
Here $\mathcal O(n)$ is necessary for the preprocessing.
For a concrete precision $\epsilon$ it makes sense to construct
the integer program and determine exactly the maximum $\ell_1$-norm
of the columns and use this to determine the size of each
hypercube $\mathsf{H}(b/2^i)$.

For example, the integer program we derive
for precision $\epsilon\approx 17.29\%$ using the optimized rounding
scheme (see next section) gives us a bound of $3$ on the $\ell_1$-norm
and therefore a very moderate bound of $5$ on the side length
of each hypercube.

\subsection*{The complete algorithm}
To sum up the description of our algorithm, we present all steps of the algorithm together in
pseudocode in Fig.~\ref{alg}.
Note that this pseudocode is not optimized (in contrast to our implementation).
For example, the computation of the sets $J_{\text{small}}$, $J_{\text{large}}$,
and $J_{\text{huge}}$ can be done in one sweep and the values $b_{i,k}$ and the
rounded processing times $\tilde{p_{j}}$ can also be computed concurrently.
The binary search is performed until lower and upper bound differ by a factor
less than $1 + \epsilon'$. This parameter $\epsilon'$ can be taken negligibly small
(in contrast to~$\epsilon$),
since the running time grows only logarithmically in $\epsilon'$.
\begin{figure}[tb]
\small{
\begin{center}
    \begin{pchstack}[boxed]
    \pseudocode[head={Input: $I=[p_{1},\ldots,p_{n},m], \varepsilon, \varepsilon'$}]{%
      p_{\max} = \max_{j=1,\ldots,n} \{p_j\}\\
      \LB = \max\{p_{\max}, \sum_{j=1}^{n} p_{j} /m \}\pccomment{compute
        bounds}\\
      L = \LB; R= 2\LB\\
      \pcwhile (1+\varepsilon')L < R \pcdo:\pccomment{binary search for $\OPT$
        }\\
      \t T=(R+L)/2 \pccomment{guess makespan}\\
      \t J_{\text{small}} = \{j\in \{1,\ldots,n\} \mid p_j \leq \varepsilon
      T\}\\
      \t J_{\text{large}} = \{1,\ldots,n\}\setminus J_{\text{small}}\\
      \t J_{\text{huge}} = \{j\in J_{\text{large}}\mid p_j \geq (1-2\epsilon)
      T\}\\
      \t J_{\text{large}} = J_{\text{large}}\setminus J_{\text{huge}}\\
      \t \pcfor j\in J_{\text{huge}}:\\
      \t[2] \text{find $j'\in J_{\text{large}}$ with $p_{j'}$ minimal and
        $p_{j}+p_{j'}\leq T$}\\
      \t[2] J_{\text{large}} = J_{\text{large}}\setminus \{j'\}\\
      \t \pcendfor\\
      \t \pcfor i=0,\ldots,\lceil  \log((1-2\varepsilon)/\varepsilon)
      \rceil:\\
      \t[2] \pcfor k=0,\ldots, \lceil  1/\varepsilon -1 \rceil:\\
      \t[3] b_{i,k} = 2^{i}\varepsilon T+k\varepsilon^2 2^i T; n_{i,k} = 0\\
      \t[2] \pcendfor\\
      \t \pcendfor
    }

    \pseudocode[head={\ }]{%
      \t \pcfor j\in J_{\text{large}}:\\
      \t[2] \text{let } i\in \mathbb N \text{ with } b_{i,k}\leq p_j < b_{i,k+1} \\ \t[5] \text{or } b_{i,\lceil1/\varepsilon -1 \rceil}\leq p_j < b_{i+1,0} \\
      \t[2] \tilde{p_j} = b_{i,k}; n_{i,k} = n_{i,k}+1\\
      \t \pcendfor \\
      \t[1] \text{construct 
        $\mathsf{IP}_{\mathcal{C}_{\textrm{red}},\hat{\mathcal{C}}}$ for
        $m-|J_{\text{huge}}|$ machines}\\
      \t \text{solve
        $\mathsf{IP}_{\mathcal{C}_{\textrm{red}},\hat{\mathcal{C}}}$ via the
        JR-algorithm}\\
      \t \pcfor j\in J_{\text{huge}}:\\
      \t[2] \text{let $j'$ be the job paired with $j$ if there is one}\\
      \t[2] \text{assign $j$ and possibly $j'$ to an empty machine}\\
      \t \pcendfor \\
      \t \text{assign $J_{\text{small}}$ greedily} \\
      \t \pcif
      \text{$\mathsf{IP}_{\mathcal{C}_{\textrm{red}},\hat{\mathcal{C}}}$ has no
        solution}:\\
      \t[2] \text{set $L=T$ and \textbf{break}}\\
      \t \pcif \text{the makespan exceeds $(1 + \varepsilon) T$}:\\
      \t[2] \text{set $L=T$ and \textbf{break}}\\
      \t \text{set $R=T$}\\
      \textbf{endwhile}\\
      \pcreturn \text{the schedule produced for $L$}
      }
    \end{pchstack}

  \end{center}}
  \caption{The non-optimized pseudocode of our algorithm. Here, \textbf{break} means that the current \textbf{while}-iteration is aborted and the next iteration is started. }
  \label{alg}
\end{figure}

\subsection*{Optimizing the rounding scheme}
Taking a closer look at the proof of Lemma~\ref{lem:reduced} reveals that we
only used the last property of Lemma~\ref{lem:rounding} to obtain the
simplified integer program.
Now, we want to find the \emph{best} rounding: Either minimize the number of
rounded processing times $d$ for a given $\varepsilon$ or minimize the precision
$\varepsilon$ for a given $d$.
Fortunately, the task to decide whether such a rounding for given $d$ and
$\varepsilon$ exists can be formulated as a mixed integer program. 
This allows us to obtain the best rounding in a generic preprocessing step.
In the following, $x_{i}$ will denote the $i$-th rounded processing time with
$x_{0}\geq x_{1}\geq \cdots \geq x_{d-1}$. 
Without loss of generality, we assume that our current guess $T$ is equal to $1$
here.
As our remaining item sizes are in the interval $(\varepsilon, 1-2\varepsilon)$
and we want to obtain a rounding that only produces an error of $1+\varepsilon$,
we need to guarantee that
\begin{align*}
  (1+\varepsilon)x_{0} &\geq (1-2\varepsilon), \\
  x_{0} &\leq (1-2\varepsilon), \\
  x_{d-1} &\leq \varepsilon (1+\varepsilon).
\end{align*}
Furthermore, we need to make sure that $x_{i}$ and $x_{i+1}$ are relatively
close, i.\,e~$(1+\varepsilon)x_{i+1} \geq x_i$. 

To guarantee the last property of Lemma~\ref{lem:rounding}, we need to ensure
that every configuration with more than $2\log(1/\varepsilon)$ 
processing times can be reduced.
Hence, we construct for every subset $X'\subseteq \{0,\ldots,d-1\}$ of 
size $L+1$ an indicator variable $y_{X'}$ that is
$1$ iff there is a single configuration  containing all of these processing
times, i.\,e.~$\sum_{i\in X'}x_{i}\leq 1$.
We also use another indicator variable $z_{i_{1},i_{2},i}$ that is $1$ iff
$x_{i_{1}}+x_{i_{2}}=x_{i}$.
Now, we need to guarantee that all configurations containing $L+1$ item sizes
can be reduced to $L$ item sizes. 
Hence, if $y_{X'}=1$ this implies that one can get rid of one of the item sizes,
i.\,e.~$z_{i_{1},i_{2},i}=1$ for some $i_{1},i_{2}\in X'$ and $i\in
\{0,\ldots,d-1\}$. 
All of these indicator variables and implications can be easily
introduced via the big $M$ method or can be directly formulated for the mixed
integer program solver.
Finally, the $x_{i}$ are fractional variables and the indicator variables are
integral.
The complete, formal formulation of the mixed integer program
$\mathsf{MIP}_{\varepsilon,d}$ can be found in Section~\ref{sec:milp} in the appendix.
Now, we can perform a binary search to optimize either $\varepsilon$ or $d$.
More formally, for a given $d$, we can find the minimal $\varepsilon(d)$, or for a
given $\varepsilon$, we can find the minimal $d(\varepsilon)$. 
See Table~\ref{tab:rounding} for the precision achievable with certain values of $d$.

\begin{table}
\caption{Optimal rounding scheme for some numbers of rounded processing times $d$.}
\label{tab:rounding}
\medskip

\centering
\begin{tabular}{c c}
  \toprule
  $d$ & optimal precision $\varepsilon(d)$ \\
  \midrule
  $9$ & $0{.}172874755859$\\
$10$ & $0{.}160867004395$\\
$11$ & $0{.}15059387207$\\
\bottomrule
\end{tabular}
\end{table}


\subsection{Non-PTAS algorithms}

Currently, heuristic algorithms like \textsf{LPT}, the \textsf{MULTIFIT} algorithm \cite{DBLP:journals/siamcomp/CoffmanGJ78} and its derivative \textsf{DJMS}~\cite{DBLP:journals/ijoris/KuruvillaP15} give some of the best algorithms to solve instances of $\text{P}||\text{C}_{\max}$ in practice. For the sake of completeness here we describe briefly how they work.

The \textsf{MULTIFIT} algorithm was presented by Coffman et al.
\cite{DBLP:journals/siamcomp/CoffmanGJ78} and its exact bound of 13/11 was
proved by Yue \cite{Yue1990OnTE}.
It is based on iteratively applying the \emph{First-Fit-Decreasing} algorithm,
which, for a given guess on the makespan $T$, sorts the jobs by their processing
times in non-increasing order and the machines in an arbitrary order.
Then, each job $j$ is packed onto the first machine where it fits i.\,e.~where
the load of $T$ is not exceeded. 
To approximate an optimum solution of an instance $I$, the \textsf{MULTIFIT} algorithm takes $I$ and a maximum number of rounds $t$ as the input and starts with computing a lower bound $\ell = \max\{\tfrac1m\sum_j p_j, p_1, p_m+p_{m+1}\}$ (where $p_1 \geq\dots\geq p_n$) and an upper bound $u = \operatorname{\textsf{LPT}}(I)$ to bound the intended makespan. Then \textsf{MULTIFIT} aims to compute the smallest feasible makespan $T^{\star} \in [\ell,u]$, which admits a First-Fit-Decreasing packing of all jobs to $m$ bins, via binary search in at most $t$ rounds. Obviously, that maximum number of rounds $t$ may be dropped if all numbers in the input are integral.
However, it can happen that $t$ rounds are over or even $T^{\star} = u$ does not
admit a feasible First-Fit-Decreasing packing. Then the algorithm returns the solution obtained by \textsf{LPT}.

The \textsf{DJMS} algorithm of Kuruvilla and
Paletta~\cite{DBLP:journals/ijoris/KuruvillaP15} combines the \textsf{LPT} approach with
the \textsf{MULTIFIT} approach by splitting machines and jobs into \emph{active} and
\emph{closed} ones.
At the beginning of the algorithm, all jobs and machines are active.
Now, in each iteration, the \textsf{LPT} algorithm is applied to the active jobs and
active machines to compute an upper bound on the makespan $T_{\text{\textsf{LPT}}}$.
Afterwards, \textsf{MULTIFIT} with the upper bound of $T_{\text{\textsf{LPT}}}$ is applied to the
active jobs and machines.
In the solution produced by \textsf{MULTIFIT}, we search for the least loaded active
machine $i^{\star}$ whose load exceeds the lower bound $\ell$ (as defined for \textsf{MULTIFIT}).
All machines with the same load as $i^{\star}$ and the jobs on them are declared
close.
These steps are repeated until all machines are closed.
In their computational experiments, \textsf{DJMS} gave the best makespan compared with
\textsf{LISTFIT} (an algorithm by Gupta and Ruiz-Torres~\cite{gupta2001listfit}), \textsf{LPT},
and \textsf{MULTIFIT}~\cite{DBLP:journals/ijoris/KuruvillaP15}.

\section{Implementation}

The algorithm was implemented in the C++ programming language. To compute fast Fourier transformations we used the C library \textsf{FFTW3} \cite{DBLP:journals/pieee/FrigoJ05} in version 3.3.8
and we applied \textsf{OpenMP} (https://www.openmp.org/) in version 5.0 for
parallelization.
We plan to upload a cleaned up version of our implementation on github
(https://www.github.com). 
The experiments were computed in the HPC Linux Cluster of Kiel university  using 16 cpu cores and 100GB of memory per instance.

\subsection{Computional results}

In the following we refer to our algorithm as \textsf{\BDJR}. To compare our implementation with the non-\ac{PTAS} algorithms \textsf{LPT}, \textsf{MULTIFIT}, and \textsf{DJMS} we compute solutions to a set of instances first considered by Kedia \cite{kedia1971job} in 1971. Since then it has been used by various authors to investigate the quality of scheduling algorithms on identical machines (cf. \cite{guptaruiztorres,DBLP:journals/ijoris/KuruvillaP15,DBLP:journals/dam/LeeM88}). The instances are grouped into four families E1, E2, E3, E4 (see \Cref{tab:instance_families} for an overview).
\begin{table*}[h]
    \caption{An overview on the instance families}
    \label{tab:instance_families}
    \medskip
    
    \centering
\begin{tabular}{cccc}
\toprule
     & $m$ & $n$ & $U$\\
     \midrule
  E1 & $3, 4, 5$ & $2m, 3m, 5m$ & $[1,20], [20,50]$\\
  \midrule
  \multirow{2}*{E2} & $2, 3$ & $10, 30, 50, 100$ & $[100,800]$\\
  & $4,6,8,10$ & $30,50,100$ & $[100,800]$\\
  \midrule
  E3 & $3,5,8,10$ & \begin{tabular}{c}$3m+1,3m+2$,\\$4m+1,4m+2$,\\$5m+1,5m+2$\end{tabular} & $[1,100], [100,200]$\\
  \midrule
  \multirow{2}*{E4} & $2$ & $10$ & $[1,20],[20,50],[1,100],[50,100],[100,200],[100,800]$\\
     & $3$ & $9$ & $[1,20],[20,50],[1,100],[50,100],[100,200],[100,800]$\\
     \bottomrule
\end{tabular}
\end{table*}
Each family consists of classes of $100$ instances and the instances of each class were generated respecting three common parameters; namely, a number of machines $m$, a number of jobs $n$, and a universe interval $U$, used to uniformly select the integer processing times of the jobs. Overall there are $90$ classes $(m,n,U)$ considered, i.e. $9000$ instances in total. These instances do not exceed a machine number of $m=10$, so we generated a family BIG of additional classes where $m \in \{25,50,75,100\}$, $n = 4m$, and $U=[1,1000]$ (see \cref{tab:BIG}) to give evidence to the fact that \textsf{\BDJR} solves larger instances in reasonable time too.
See \Cref{fig:M3_N9_U1_20_new} for a makespan comparison of \textsf{\BDJR} with the non-\ac{PTAS} algorithms for class $(m=4,n=8,U=[1,20])$ of family E1 where the 100 instances of the class are presented from left to right while the makespan grows from bottom to top.
In \cref{tab:E1,tab:E2,tab:E3,tab:E4,tab:BIG} we give the computation results which are prepared as follows. Note that the results for families E3, E4, and BIG can be found in the appendix in Section~\ref{apx:results}.
Each line summarizes the results for the $100$ instances of class $(m,n,U)$. Column \texttt{better} is the number of instances where the makespan computed by {\BDJR} is lower than the best makespan of the non-\ac{PTAS} algorithms while \texttt{equal} counts the instances where these are equal. Column $\texttt{avg\_quot} = \sum_{i=1}^{100} \textsf{\BDJR}(i) / \min\{\sum_{i=1}^{100} \textsf{LPT}(i), \sum_{i=1}^{100} \textsf{MF}(i), \sum_{i=1}^{100} \textsf{DJMS}(i)\}$ compares the sums of the computed makespans (and is rounded to two decimal places) and column \texttt{avg\_time} states the rounded-up average running time of {\BDJR} in minutes. 
There were no instance classes, where the maximal running time exceeded $2\cdot \texttt{avg\_time}$. 
\begin{table}[h]
\caption{Computational results for the classes of family E1}
\label{tab:E1}
\centering
\scriptsize
\begin{tabular}{rrrrrrcc}
\toprule
family & $m$ & $n$ & $U$ & better & equal & avg\_quot & avg\_time\\
\midrule
E1 & $3$ & $6$ & $[1,20]$ & $0$ & $69$ & $1.02$ & $24$\\
E1 & $3$ & $6$ & $[20,50]$ & $0$ & $22$ & $1.05$ & $44$\\
E1 & $3$ & $9$ & $[1,20]$ & $2$ & $35$ & $1.03$ & $35$\\
E1 & $3$ & $9$ & $[20,50]$ & $3$ & $7$ & $1.05$ & $57$\\
E1 & $3$ & $15$ & $[1,20]$ & $2$ & $42$ & $1.02$ & $42$\\
E1 & $3$ & $15$ & $[20,50]$ & $2$ & $2$ & $1.04$ & $49$\\
E1 & $4$ & $8$ & $[1,20]$ & $0$ & $48$ & $1.03$ & $29$\\
E1 & $4$ & $8$ & $[20,50]$ & $0$ & $11$ & $1.05$ & $49$\\
E1 & $4$ & $12$ & $[1,20]$ & $1$ & $25$ & $1.04$ & $38$\\
E1 & $4$ & $12$ & $[20,50]$ & $0$ & $1$ & $1.06$ & $60$\\
E1 & $4$ & $20$ & $[1,20]$ & $1$ & $23$ & $1.05$ & $46$\\
E1 & $4$ & $20$ & $[20,50]$ & $0$ & $1$ & $1.06$ & $65$\\
E1 & $5$ & $10$ & $[1,20]$ & $0$ & $50$ & $1.03$ & $38$\\
E1 & $5$ & $10$ & $[20,50]$ & $0$ & $5$ & $1.06$ & $60$\\
E1 & $5$ & $15$ & $[1,20]$ & $0$ & $8$ & $1.06$ & $47$\\
E1 & $5$ & $15$ & $[20,50]$ & $0$ & $1$ & $1.07$ & $69$\\
E1 & $5$ & $25$ & $[1,20]$ & $0$ & $7$ & $1.06$ & $51$\\
E1 & $5$ & $25$ & $[20,50]$ & $0$ & $0$ & $1.07$ & $72$\\
\bottomrule
\end{tabular}
\end{table}
\begin{table}[h!]
\caption{Computational results for the classes of family E2}
\label{tab:E2}
\centering
\scriptsize
\begin{tabular}{rrrrrrcc}
\toprule
family & $m$ & $n$ & $U$ & better & equal & avg\_quot & avg\_time\\
\midrule
E2 & $2$ & $10$ & $[100,800]$ & $16$ & $18$ & $1.01$ & $99$\\
E2 & $2$ & $30$ & $[100,800]$ & $0$ & $58$ & $1$ & $1$\\
E2 & $2$ & $50$ & $[100,800]$ & $0$ & $56$ & $1$ & $1$\\
E2 & $2$ & $100$ & $[100,800]$ & $0$ & $66$ & $1$ & $1$\\
E2 & $3$ & $10$ & $[100,800]$ & $10$ & $7$ & $1.03$ & $125$\\
E2 & $3$ & $30$ & $[100,800]$ & $2$ & $29$ & $1$ & $1$\\
E2 & $3$ & $50$ & $[100,800]$ & $0$ & $0$ & $1$ & $1$\\
E2 & $3$ & $100$ & $[100,800]$ & $0$ & $0$ & $1$ & $1$\\
E2 & $4$ & $30$ & $[100,800]$ & $7$ & $0$ & $1.03$ & $130$\\
E2 & $4$ & $50$ & $[100,800]$ & $0$ & $0$ & $1.01$ & $1$\\
E2 & $4$ & $100$ & $[100,800]$ & $0$ & $38$ & $1$ & $1$\\
E2 & $6$ & $30$ & $[100,800]$ & $1$ & $0$ & $1.06$ & $145$\\
E2 & $6$ & $50$ & $[100,800]$ & $6$ & $0$ & $1.02$ & $136$\\
E2 & $6$ & $100$ & $[100,800]$ & $0$ & $0$ & $1$ & $1$\\
E2 & $8$ & $30$ & $[100,800]$ & $0$ & $0$ & $1.08$ & $167$\\
E2 & $8$ & $50$ & $[100,800]$ & $0$ & $0$ & $1.07$ & $162$\\
E2 & $8$ & $100$ & $[100,800]$ & $0$ & $0$ & $1.01$ & $1$\\
E2 & $10$ & $30$ & $[100,800]$ & $0$ & $0$ & $1.08$ & $171$\\
E2 & $10$ & $50$ & $[100,800]$ & $0$ & $0$ & $1.09$ & $164$\\
E2 & $10$ & $100$ & $[100,800]$ & $0$ & $7$ & $1.01$ & $25$\\
\bottomrule
\end{tabular}
\end{table}

For each of the classes where $\texttt{avg\_time} = 1$ the short running time is explained by the rather large quotient $n/m$ causing trivial instances where nearly all jobs are either small ($\leq\varepsilon$) or huge ($\geq 1-2\varepsilon$). 

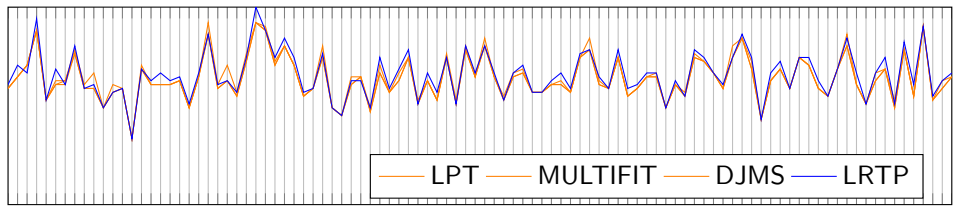
\begin{figure}[h!]
    \centering
    \pgfplotstableread{results/M3_N9_U1_20_new.dat}{\MTHREExNNINExUONExTWENTY}
    \begin{tikzpicture}
        \begin{axis}[
            xmin = 1, xmax = 100,
            ymin = 0, ymax = 1,
            xtick distance = 1,
            ytick distance = 1,
            xmajorgrids=true,
            major grid style = {lightgray},
            width = \textwidth,
            height = 0.3\textwidth,
            legend cell align = {left},
            legend pos = south east,
            legend columns = 4,
            xticklabels={,,},
            yticklabels={,,}
        ]
        \addplot[orange] table [x = {I}, y = {LPT}] {\MTHREExNNINExUONExTWENTY};
        \addplot[orange] table [x ={I}, y = {MF}] {\MTHREExNNINExUONExTWENTY};
        \addplot[orange] table [x = {I}, y = {DJMS}] {\MTHREExNNINExUONExTWENTY};
        \addplot[blue] table [x = {I}, y = {BDJR}] {\MTHREExNNINExUONExTWENTY};
        \legend{\textsf{LPT},\textsf{MULTIFIT},\textsf{DJMS},\textsf{\BDJR}}
        \end{axis}
    \end{tikzpicture}
    \caption{Makespan comparison for 100 instances of class $(3,8,[1,20])$ (x: instances, y: makespan)}
    \label{fig:M3_N9_U1_20_new}
\end{figure}

By their simple nature the non-\ac{PTAS} algorithms
compute solutions even to large instances within a few seconds.
While the running time of the \ac{PTAS}
is generally on a high level, the increase in the running time
as the number of machines and jobs grows is rather small
due to the parameterized running time of our algorithm.
We emphasize that the running time did not cross the border of $5.6$ hours for any instance computed.
We therefore see these experiments as a valid proof of concept for
practically feasible \ac{PTAS}'s.
The empirical solution quality, although not our main focus, is not superior to the
non-\ac{PTAS} algorithms at the considered precision, though the difference is only small.
This may be due to the random instances which do not necessarily
exhibit a worst-case structure and that the difference in theoretical approximation guarantee
is at this state fairly small.
We believe that with additional computational resources or further optimizations an
even lower precision can be reached, which will lead to a superior solution quality also
in experiments.

\bibliography{myref}

\appendix

\section{Omitted proofs}
\LargeAndSmallJobs*
\begin{proof}
  There are two possible situations to consider: 
  If the makespan does not increase, we have
  $\mu(\delta)=\mu(\delta_{\text{large}})$. 
  If the makespan increases, consider the loads $\ell_{1},\ldots,\ell_{m}$ of
  the machines. 
  As the makespan increased, the differences between the loads is bounded by
  $\varepsilon T$, i.\,e.~$|\ell_{i}-\ell_{i'}|\leq \varepsilon T$ for all $i,i'$. 
  We thus have $\mu(\delta)-\ell_{i}\leq \varepsilon T$ for all $i$, as
  $\mu(\delta)=\max_{i}\ell_{i}$. Hence
  \begin{align*}
    \sum_{j}p_{j} = \sum_{i} \ell_{i} \geq \sum_{i} (\mu(\delta)-\varepsilon T)= m\cdot (\mu(\delta)-\varepsilon T)
  \end{align*}
  and thus $(\sum_{j} p_{j})/m \geq \mu(\delta)-\varepsilon T$. As $\OPT(J)\geq
  (\sum_{j} p_{j})/m$, 
  we conclude  $\mu(\delta)\leq \OPT(J)+\varepsilon T$.
  The time complexity of $\mathcal{O}((|J_{\text{small}}|+m)\cdot
  \log(|J_{\text{small}}|+m))$ can be achieved by storing the loads of the
  machines in a min-heap. 
\end{proof}

\begin{lemma}[Formal version of Lemma~\ref{lem:huge_jobs}]
  Define $J_{\text{medium}}:=\{j\in J\mid \varepsilon T < p_{j}\leq 2
  \varepsilon T\}\subseteq J_{\text{large}}$. 
  Let $J_{\text{huge}}=\{j_{1},\ldots,j_{k}\}$ with $p_{j_{1}}\geq p_{j_{2}}\geq
  \ldots p_{j_{k}}$ and $J_{\text{medium}}=\{j'_{1},\ldots,j'_{k'}\}$ with
  $p_{j'_{1}}\geq p_{j'_{2}}\geq \ldots p_{j'_{k'}}$. 
  Define iteratively $\psi\colon \{1,\ldots,k\}\to \{1,\ldots,k',\infty\}$ with
  $\psi(j)=\min\{j'\leq k'\mid p_{j}+p_{j'}\leq 1\land j'\not\in \{\psi(\tilde{j}) \mid \tilde{j} < j\}\}$, where $\min \emptyset=\infty$.
  
  There is an optimal schedule $\delta_{\text{large}}$ of $J_{\text{large}}$
  such that for each machine $i$ with $\delta_{\text{large}}^{-1}(i)\cap
  J_{\text{huge}}=\{j\}$, we have either  $\delta_{\text{large}}^{-1}(i)=\{j\}$
  if $\psi(j)=\infty$ or $\delta_{\text{large}}^{-1}(i)=\{j,\psi(j)\}$ if
  $\psi(j)< \infty$. 
\end{lemma}

\begin{proof}
  Intuitively, $\psi$ maps
  $j$ to the largest job $j'$ in $J_{\text{medium}}$ with which it can be put onto a
  machine and which is not already mapped to another huge job. If no such job $j'$
  exists, $\psi(j)=\infty$. 

  Consider any optimal schedule $\delta_{\text{large}}$ of $J_{\text{large}}$
  such that there is a machine $i$ with $\delta_{\text{large}}^{-1}(i)\cap
  J_{\text{huge}}=\{j\}$ with $\delta_{\text{large}}^{-1}(i)=\{j,j'\}$ where
  either $\psi(j)=\infty$ or $j'\neq \psi(j)$. 
  If there are multiple such machines, we consider the one with the job $j=j_{\ell}$ with
  minimal index in the ordering $j_{1},\ldots,j_{k}$.
  As $j_{\ell}$ has minimal index, if $\psi(j)=\infty$, all fitting jobs from
  $J_{\text{medium}}$ are packed with the jobs $j_{1},\ldots,j_{\ell-1}$. Hence,
  there is no job $j'$ that can be packed with $j_{\ell}$. This situation is
  thus not possible.
  If $j'\neq \psi(j)$, consider the machine
  $\delta_{\text{large}}(\psi(j)):=i'$. As $p_{\psi(j)}+p_{j}\leq 1$ and
  $p_{\psi(j)}\geq p_{j'}$ (by definition of $\psi(j)$), we can exchange $j'$
  and $\psi(j)$ and still keep a feasible optimal schedule.

  Applying the above reason iteratively finally gives us an optimal schedule
  that always pairs the jobs $j$ and $\psi(j)$ (if $\psi(j) < \infty$). 
\end{proof}

\begin{lemma}[Formal version of Lemma~\ref{lem:rounding}]
  \begin{enumerate}
  \item The number of different rounded processing times of $J_{\text{rem}}$ is
    at most $\lceil  1/\epsilon \rceil\cdot
    [\log((1-2\epsilon)/\epsilon)+1]$. 
  \item A schedule $\tilde{\sigma}$ of the rounded processing times implies a
    schedule $\sigma$ of the original processing time with $\mu(\sigma)\leq
    (1+\varepsilon)\mu(\tilde{\sigma})$. 
  \item We have $b_{i,k_{1}}+b_{i,k_{2}}=b_{i+1,(k_{1}+k_{2})/2}$ for all $i$ and
    all $k_{1},k_{2}$ with $k_{1}\bmod 2 = k_{2}\bmod 2$. 
  \end{enumerate}
\end{lemma}

\begin{proof}\
  
     \begin{enumerate}
  \item We have $2^{i+1}\epsilon T = (1-2\epsilon)T$ iff $i+1=\log((1-2\epsilon)/\epsilon)$.
    For $i > \log((1-2\epsilon)/\epsilon)$, there is no job in any interval $I_{i}$.
    Hence, there are at most $\log((1-2\epsilon)/\epsilon)+1$ many
    values for $i$ such that the processing times lie in $I_{i}$. 
    Every $I_{i}$ is split into $\lceil  1/\epsilon \rceil+1$ intervals. Hence, there are at most $[\lceil  1/\epsilon \rceil+1]\cdot [\log((1-2\epsilon)/\epsilon)+1]$ many
    rounded processing times. 
  \item It is sufficient to show $p_{j}\geq \tilde{p}_{j}\geq
    (1+\epsilon)^{-1}p_{j}$. As we only round down, we have $p_{j}\geq
    \tilde{p}_{j}$ immediately. 
    Let  $\tilde{p}_{j}=b_{i,k}$ and hence $p_{j}\in I_{i,k}$.
    Then
    \begin{align*}
      &(1+\epsilon)\tilde{p}_{j}= (1+\epsilon)b_{i,k}=(1+\epsilon)[2^{i}\epsilon T+k\epsilon^{2}2^{i}T]=\\
      &2^{i}\epsilon T+k\epsilon^{2}2^{i}T+2^{i}\epsilon^{2}T+k\epsilon^{3}2^{i}T \geq\\
      &2^{i}\epsilon T+k\epsilon^{2}2^{i}T+2^{i}\epsilon^{2}T=2^{i}\epsilon T+(k+1)\epsilon^{2}2^{i}T = b_{i,k+1} > p_{j}. 
    \end{align*}
  \item We have 
    \begin{align*}
      b_{i,k_{1}}+b_{i,k_{2}}=2^{i}\epsilon T+k_{1}\epsilon^{2}2^{i}T+2^{i}\epsilon T+k_{2}\epsilon^{2}2^{i}T =
      2^{i+1}\epsilon T+(k_{1}+k_{2})\epsilon^{2}2^{i}T.
    \end{align*}
   As $k_{1}\bmod 2 = k_{2}\bmod 2$, we have $(k_{1}+k_{2})/2\in \mathbb{Z}_{\geq
      0}$ and thus 
    \begin{align*}
      2^{i+1}\epsilon T+(k_{1}+k_{2})\epsilon^{2}2^{i}T=
      2^{i+1}\epsilon T+[(k_{1}+k_{2})/2]\epsilon^{2}2^{i+1}T=b_{i+1,(k_{1}+k_{2})/2}.
    \end{align*}
  \end{enumerate}
  
\end{proof}

\paragraph*{Proving Lemma~\ref{lem:reduced}.}
To prove Lemma~\ref{lem:reduced}, we first introduce some notations. 
For $i^{*},k^{*}_{1},k^{*}_{2}$ with $k^{*}_{1}\bmod 2 = k^{*}_{2}\bmod 2$ and,
we define the new column $\hat{c}:=\hat{c}(i^{*},k^{*}_{1},k^{*}_{2})\in \hat{\mathcal{C}}$
with (i) $\hat{c}[i,k] = 2$, if $i=i^{*}$ and $k=k^{*}_{1}=k^{*}_{2}$ or (ii)
$\hat{c}[i,k] = 1$, if  $i=i^{*}$ and $k\in \{k^{*}_{1},k^{*}_{2}\}$ and
$k^{*}_{1}\neq k^{*}_{2}$ or (iii) $\hat{c}[i,k]=-1$, if  $i=i^{*}+1$ und
$k=(k^{*}_{1}+k^{*}_{2})/2$. In all other cases, we define $\hat{c}[i,k]=0$.

For a subset  $\mathcal{C}'\subseteq \mathcal{C}$, let 
$\mathsf{IP}_{\mathcal{C}',\hat{\mathcal{C}}}$ be the following integer program: 
\begin{align*}
  &\sum_{c\in \mathcal{C}'}x_{c} \leq m;\quad 
  \sum_{c\in \mathcal{C}'\cup \hat{\mathcal{C}}}x_{c}\cdot c[i,k] \geq n_{i,k} \ \ \forall (i,k); \quad
  x_{c}\in \mathbb{Z}_{\geq 0} \ \  \forall c\in \mathcal{C}'\cup \hat{\mathcal{C}}
\end{align*}

Let $\mathsf{IP}_{\mathcal{C}}$ be the original configuration \ac{IP}.
We can now obtain the following switching lemma, which directly implies
Lemma~\ref{lem:reduced}, as a feasible solution $x^{*}$ of
$\mathsf{IP}_{\mathcal{C}}$ can be transformed into a feasible solution $y^{*}$
of $\mathsf{IP}_{\mathcal{C}_{\textrm{red}},\hat{\mathcal{C}}}$ and vice versa. 

\begin{lemma}
  Let $x^{*}$ be a feasible solution of 
  $\mathsf{IP}_{\mathcal{C},\hat{\mathcal{C}}}$. 
  \begin{enumerate}
  \item If there is a $c_{1}\in \mathcal{C}$ with $\vert\vert
    c_1\vert\vert_1 > 2[\log((1-2\epsilon)/\epsilon)+1]$ and $x^{*}_{c_{1}}>0$, there are
    configurations 
    $c_{2}\in \mathcal{C}$ with $\vert\vert c_{2}\vert\vert_{1} < \vert\vert
    c_{1}\vert\vert_{1}$ and $c_{3}\in \hat{\mathcal{C}}$, such that the
    vector $y$ is a feasible solution of 
    $\mathsf{IP}_{\mathcal{C},\hat{ \mathcal{C}}}$ with  $y_{c_{1}}=x^{*}_{c_{1}}-1$,
    $y_{c_{2}}=x^{*}_{c_{2}}+1$, $y_{c_{3}}=1$ and
    $y_{c}=x^{*}_{c}$ for all other $c$. 
      \item If there is a configuration $c_{1}=\hat{c}(i^{*},k^{*}_{1},k^{*}_{2})\in
    \hat{\mathcal{C}}$ with $x^{*}_{c_{1}}>0$ and a configuration $c_{2}\in
    \mathcal{C}$ with $c_{2}[i^{*}+1,(k^{*}_{1}+k^{*}_{2})/2]=1$
    and $x^{*}_{c_{2}}> 0$, there is a configuration $c_{3}\in
    \mathcal{C}$, such that the vector  $y$  is a feasible solution of 
    $\mathsf{IP}_{\mathcal{C},\hat{\mathcal{C}}}$ with 
    $y_{c_{1}}=x^{*}_{c_{1}}-1$, $y_{c_{2}}=x^{*}_{c_{2}}-1$,
    $y_{c_{3}}=x^{*}_{c_{3}}+1$ and $y_{c}=x^{*}_{c}$ for all other 
    $c$. 
      \end{enumerate}
    \end{lemma}

\begin{proof}
Let $x^{*}$ be a feasible solution of 
$\mathsf{IP}_{\mathcal{C},\hat{\mathcal{C}}}$.
\begin{enumerate}
\item As $\vert\vert c_1\vert\vert_1 > 2[\log((1-2\epsilon)/\epsilon)+1]$, the pidgeonhole
  principle implies that there are indices 
  $i^{*},k_{1}^{*},k_{2}^{*}$ with $k_{1}^{*}\bmod 2=k_{2}^{*}\bmod 2$, such
  that 
  $\min\{c_{1}[i^{*},k_{1}^{*}],c_{1}[i^{*},k_{2}^{*}]\}\geq 1$, if 
  $k_{1}^{*}\neq k_{2}^{*}$ or $c_{1}[i^{*},k_{1}^{*}]\geq 2$, if 
  $k_{1}^{*}=k_{2}^{*}$. 
  Choose $c_{3}=\hat{c}(i^{*},k_{1}^{*},k_{2}^{*})$ and
  $c_{2}=c_{1}-c_{3}$. 
  By construction, we have $c_{2}\in \mathcal{C}$ and $\lVert c_{2}
  \rVert_{1}=\lVert c_{1} \rVert_{1}+1-2 = \lVert c_{1} \rVert_{1}-1$.
  For $i=i^{*}$ and  $k=k^{*}_{1}=k^{*}_{2}$, we have 
  \begin{align*}
    \sum_{c\in \mathcal{C}'\cup \hat{\mathcal{C}}}y_{c}\cdot c[i,k] = \left[  \sum_{c\in \mathcal{C}'\cup \hat{\mathcal{C}}}x^{*}_{c}\cdot c[i,k]  \right]\underbrace{-2}_{\text{by $c_{1}$}} \underbrace{+2}_{\text{by $c_{3}$}} \geq n_{i,k}.
  \end{align*}
  For $i=i^{*}$ and $k^{*}_{1}\neq k^{*}_{2}$ and $k\in
  \{k^{*}_{1},k^{*}_{2}\}$, we have 
  \begin{align*}
    \sum_{c\in \mathcal{C}'\cup \hat{\mathcal{C}}}y_{c}\cdot c[i,k] = \left[  \sum_{c\in \mathcal{C}'\cup \hat{\mathcal{C}}}x^{*}_{c}\cdot c[i,k]  \right]\underbrace{-1}_{\text{by $c_{1}$}} \underbrace{+1}_{\text{by $c_{3}$}} \geq n_{i,k}.
  \end{align*}
  For $i=i^{*}+1$, and $k=(k^{*}_{1}+k^{*}_{2})/2$, we have 
  \begin{align*}
    \sum_{c\in \mathcal{C}'\cup \hat{\mathcal{C}}}y_{c}\cdot c[i,k] = \left[ \sum_{c\in \mathcal{C}'\cup \hat{\mathcal{C}}}x^{*}_{c}\cdot c[i,k]   \right] \underbrace{+1}_{\text{by $c_{2}$}} \underbrace{-1}_{\text{by $c_{3}$}} \geq n_{i,k}. 
  \end{align*}
  Finally, for all other $i$ and $k$, we have
  \begin{align*}
    \sum_{c\in \mathcal{C}'\cup \hat{\mathcal{C}}}y_{c}\cdot c[i,k]  = \sum_{c\in \mathcal{C}'\cup \hat{\mathcal{C}}}x^{*}_{c}\cdot c[i,k]  \geq n_{i,k},
  \end{align*}
  as nothing changed here.

  As $\sum_{c\in \mathcal{C}}y_{c} = \sum_{c\in \mathcal{C}}x^{*}_{c} \leq m$,
  we can conclude that $y$ is a feasible solution of 
  $\mathsf{IP}_{\mathcal{C},\hat{\mathcal{C}}}$.
  
\item Choose $c_{3}=c_{2}-c_{1}$. By construction, we have $c_{3}\in
  \mathcal{C}$.

    For $i=i^{*}$ and  $k=k^{*}_{1}=k^{*}_{2}$, we have 
  \begin{align*}
    \sum_{c\in \mathcal{C}'\cup \hat{\mathcal{C}}}y_{c}\cdot c[i,k] = \left[  \sum_{c\in \mathcal{C}'\cup \hat{\mathcal{C}}}x^{*}_{c}\cdot c[i,k]  \right]\underbrace{+2}_{\text{by $c_{1}$}} \underbrace{-2}_{\text{by $c_{3}$}} \geq n_{i,k}.
  \end{align*}
  For $i=i^{*}$ and $k^{*}_{1}\neq k^{*}_{2}$ and $k\in
  \{k^{*}_{1},k^{*}_{2}\}$, we have 
  \begin{align*}
    \sum_{c\in \mathcal{C}'\cup \hat{\mathcal{C}}}y_{c}\cdot c[i,k] = \left[  \sum_{c\in \mathcal{C}'\cup \hat{\mathcal{C}}}x^{*}_{c}\cdot c[i,k]  \right]\underbrace{+1}_{\text{by $c_{1}$}} \underbrace{-1}_{\text{by $c_{3}$}} \geq n_{i,k}.
  \end{align*}
  For $i=i^{*}+1$, and $k=(k^{*}_{1}+k^{*}_{2})/2$, we have 
  \begin{align*}
    \sum_{c\in \mathcal{C}'\cup \hat{\mathcal{C}}}y_{c}\cdot c[i,k] = \left[ \sum_{c\in \mathcal{C}'\cup \hat{\mathcal{C}}}x^{*}_{c}\cdot c[i,k]   \right] \underbrace{-1}_{\text{by $c_{2}$}} \underbrace{+1}_{\text{by $c_{3}$}} \geq n_{i,k}. 
  \end{align*}
  Finally, for all other $i$ and $k$, we have
  \begin{align*}
    \sum_{c\in \mathcal{C}'\cup \hat{\mathcal{C}}}y_{c}\cdot c[i,k]  = \sum_{c\in \mathcal{C}'\cup \hat{\mathcal{C}}}x^{*}_{c}\cdot c[i,k]  \geq n_{i,k},
  \end{align*}
  as nothing changed here. 

  As $\sum_{c\in \mathcal{C}}y_{c} = \sum_{c\in \mathcal{C}}x^{*}_{c} \leq m$,
  we can conclude that $y$ is a feasible solution of 
  $\mathsf{IP}_{\mathcal{C},\hat{\mathcal{C}}}$.
\end{enumerate}
\end{proof}

\section{The complete MILP to optimize the rounding scheme}
\label{sec:milp}

The $y$-variables can be easily introduced via the big $M$ method
and by minimizing the $y$ variables. 
\begin{align*}
-\sum_{i\in X'}x_{i}+M\cdot y_{X'}&\leq -1 &\forall X'\subseteq \{0,\ldots,d-1\}\land |X'|=L+1\\
( y_{X'} = 1 &\Leftrightarrow  \sum_{i\in X'}x_{i}\leq 1 )
\end{align*}
We also construct another indicator variable $z_{i_{1},i_{2},i}$ that is $1$ iff
$x_{i_{1}}+x_{i_{2}}=x_{i}$ holds. 
\begin{align*}
   z_{i_1,i_2,i} = 1 &\Leftrightarrow x_{i_1}+x_{i_2} = x_i & \forall i_1,i_2,i\in \{0,\ldots,d-1\}
\end{align*}
Again, this implication can be formulated via the big $M$ method or directly by
a mixed integer program solver.

Now, we need to guarantee that all configurations containing $L+1$ item sizes
can be reduced to $L$ item sizes. 
Hence, if $y_{X'}=1$ this implies that one can get rid of one of the item sizes,
i.\,e.~$z_{i_{1},i_{2},i}=1$ for some $i_{1},i_{2}\in X'$ and $i\in
\{0,\ldots,d-1\}$: 
\begin{align*}
  y_{X'} = 1 &\implies \exists i'_1,i'_2\in X' \exists i\in \{0,\ldots,r-1\} :
 z_{i'_1,i'_2,i} = 1 & \forall X'\subseteq \{0,\ldots,d-1\}\land |X'|=L+1
\end{align*}
Again, this implication can be formulated via the big $M$ method or directly by
an mixed integer program solver.
Finally, the $x_{i}$ are fractional variables and the indicator variables are
integral:
\begin{align*}
  x_i &\in [0,1] & \forall i\in \{0,\ldots,d-1\}\\
  y_{X'} &\in \{0,1\} &\forall X'\subseteq \{0,\ldots,d-1\}\land |X'|=L+1\\
  z_{i_1,i_2,i} &\in \{0,1\} & \forall i_1,i_2,i\in \{0,\ldots,d-1\}  
\end{align*}
In total, we obtain the following mixed integer program:
\begin{align*}
  \min &\sum_{X'}y_{X'}+\sum_{i_{1},i_{2},i}z_{i_{1},i_{2},i}\text{ s.t.}\\
  (1+\varepsilon)x_{0} &\geq (1-2\varepsilon)\\
  x_{0} &\leq (1-2\varepsilon)\\
  x_{d-1} &\leq \varepsilon\cdot (1+\varepsilon)\\
  (1+\varepsilon)x_{i+1} &\geq x_i & \forall i=1,\ldots,d-2\\
  y_{X'} = 1 &\Leftrightarrow  \sum_{i\in X'}x_{i}\leq 1\\
  z_{i_1,i_2,i} = 1 &\Leftrightarrow x_{i_1}+x_{i_2} = x_i & \forall i_1,i_2,i\in \{0,\ldots,d-1\}\\
  y_{X'} = 1 &\Rightarrow \exists i'_1,i'_2\in X' \exists i\in \{1,\ldots,n\} :
  z_{i'_1,i'_2,i} = 1 & \forall X'\subseteq \{0,\ldots,d-1\}\land |X'|=L+1\\
  x_i &\in [0,1] & \forall i\in \{0,\ldots,d-1\}\\
y_{X'} &\in \{0,1\} &\forall X'\subseteq \{0,\ldots,d-1\}\land |X'|=L+1\\
z_{i_1,i_2,i} &\in \{0,1\} & \forall i_1,i_2,i\in \{0,\ldots,d-1\}  
\end{align*}

\section{Computing Multi-Dimensional Convolutions with FFT}
\label{apx:fft}

The JR-algorithm depends on multiplying
multi-variate polynomials. Here we describe how to do this using the well-known \acf{FFT}.
We start by introducing multi-dimensional polynomials as well as multi-dimensional \acfp{DFT}.

\begin{definition}
A \emph{multivariate} or \emph{multi-dimensional} polynomial $p$ of $d$ variables $x_1,\dots,x_d$ and coordinate degree
$n = (n_1,\dots,n_d)$ is a linear combination of monomials, i.e.
\[
p(x) = \sum_{k\leq n-\mathbf{1}} p_k x^k
\qquad \text{where} \quad
x^k := \prod_{j=1}^{d} x_j^{k_j} \quad \forall k \in \{0,\dots,n_1-1\} \times \dots \times \{0,\dots,n_d-1\},
\]
$n-\mathbf{1} = (n_1-1,\dots,n_d-1)$, and $p_k \in \mathbb{C}$ f.a. $k \leq n-\mathbf{1}$.
So $p$ is a (univariate) polynomial of degree $n_j$ in each coordinate direction $j$.
\end{definition}
\begin{definition}
For two multivariate polynomials $f(x) = \sum_{k \leq n-\mathbf{1}} f_k x^k$, $g(x) = \sum_{k \leq n-\mathbf{1}}g_k x^k$ the multi-dimensional discrete convolution $(f * g)(x) = \sum_{k \leq 2n-\mathbf{2}} c_k x^k$ of $f$ and $g$ is defined by
\[
c_k = \sum_{j \leq k} f_j \cdot g_{k-j}, \quad \text{i.e.} \quad
c_{(k_1,\dots,k_d)}
=\sum_{j_1=0}^{k_1} \sum_{j_2=0}^{k_2} \dots \sum_{j_d=0}^{k_d} f_{(j_1,\dots,j_d)} \cdot g_{(k_1-j_1,\dots,k_d-j_d)}.
\]
\end{definition}
Hence, a truly primitive approach to compute $f * g$ takes time $\mathcal{O}(\sum_{k \leq 2n-\mathbf{2}} \prod_{j=1}^d (k_j+1)) \leq \mathcal{O}(2^dN^2)$ where $N = \prod_{j=1}^d n_j$ is the number of all input points.

However, the convolution theorem says that $\mathcal{F}\{f * g\} = N \cdot \mathcal{F}\{f\} \odot \mathcal{F}\{g\}$ where $\mathcal{F}$ denotes the Fourier transform operator and $\odot$ denotes the point-wise multiplication, i.e. $(v \odot w)_k = v_k \cdot w_k$. Therefore, one can compute the convolution $f * g$ by computing \[f * g = \mathcal{F}^{-1}\{N \cdot \mathcal{F}\{f\} \odot \mathcal{F}\{g\}\}.\]
For our goals we only depend on the \emph{discrete} Fourier transformation as follows.
\begin{definition}[\ac{DFT}]\label{dft}
    The multi-dimensional \acf{DFT} $\hat{f} = \mathcal{F}\{f\}$ of $f$ is defined by
    \[\hat{f}_k = \sum_{\ell \leq n-\mathbf{1}} f_{\ell} \cdot \exp\left(-i2\pi\sum_{j=1}^d\frac{k_j\ell_j}{n_j}\right) \quad \forall k \leq n-\mathbf{1}\]
    whereas the \emph{inverse} multi-dimensional \ac{DFT} $\mathcal{F}^{-1}\{\hat{f}\}$ of $\hat{f}$ is given by
    \[f_k = \frac1N\sum_{\ell \leq n-\mathbf{1}} \hat{f}_{\ell} \cdot \exp\left(i2\pi\sum_{j=1}^d\frac{k_j\ell_j}{n_j}\right) \quad \forall k \leq n-\mathbf{1}.\]
\end{definition}
To give more light to these definitions let us consider the $1$-dimensional case, i.e. $d=1, n = n_1$. Then Definition~\ref{dft} simplifies to
\[\hat{f}_k = \sum_{\ell=0}^{n-1}f_{\ell} \cdot e^{-i2\pi k\ell/n}
\quad \text{and} \quad
f_k = \frac1n\sum_{\ell=0}^{n-1}\hat{f}_\ell \cdot e^{i2\pi k\ell/n} \quad \forall k=0,\dots,n-1.\]
Obviously, since $\hat{f}_k = f(e^{-i2\pi k/n})$ for $k=0,\dots,n-1$ one can compute the \ac{DFT} of $f$ in time $\mathcal{O}(n^2)$ by simply evaluating $f$ in $n$ points. Fortunately, this running time can be reduced to $\mathcal{O}(n\log n)$ by algorithmically exploiting the fact, that the evaluation of
\[
    f(x) \; = \; (f_0 x^0 + f_2 x^2 + \dots) + x\cdot(f_1 x^0 + f_3 x^2 + \dots) \; = \; p(x^2) + x\cdot q(x^2)
\]
can always be written by the evaluation of two polynomials $p,q$ of half-sized degrees for any input point $x$. This running time improvement is the reason to call a \ac{DFT} an \ac{FFT}. Also the inverse can be computed efficiently by using just the same trick, since $f_k = \frac1n \hat{f}(e^{i2\pi k/n})$. See Fig. \ref{fig:DFTvsFFT} for an illustration of the whole computation.

\begin{figure}[h]
    \centering
    \begin{tikzpicture}[xscale=0.95]
        \usetikzlibrary{shapes}
        \node at (2.25,3) {\ac{DFT}};
        \node [draw=gray,minimum width=2.5cm] (fg) at (0,3) {$f,g$};
        \node [draw=gray,minimum width=2.5cm] (FfFg) at (0,0) {$\mathcal{F}\{f\},\mathcal{F}\{g\}$};
        \node [draw=gray,minimum width=2.5cm] (FfoFg) at (4.5,0) {$\mathcal{F}\{f\}\odot \mathcal{F}\{g\}$};
        \node [draw=gray,minimum width=2.5cm] (f*g) at (4.5,3) {$f * g$};
        \draw [->] (fg) -- node [right] {$2\cdot\mathcal{O}(n^2)$} (FfFg);
        \draw [->] (FfFg) -- node [above] {$\mathcal{O}(n)$} (FfoFg);
        \draw [->] (FfoFg) -- node [left] {$\mathcal{O}(n^2)$} (f*g);
        
        \node at (9.75,3) {\ac{FFT}};
        \node [draw=gray,minimum width=2.5cm] (fg) at (7.5,3) {$f,g$};
        \node [draw=gray,minimum width=2.5cm] (FfFg) at (7.5,0) {$\mathcal{F}\{f\},\mathcal{F}\{g\}$};
        \node [draw=gray,minimum width=2.5cm] (FfoFg) at (12,0) {$\mathcal{F}\{f\}\odot \mathcal{F}\{g\}$};
        \node [draw=gray,minimum width=2.5cm] (f*g) at (12,3) {$f * g$};
        \draw [->] (fg) -- node [right] {$2\cdot\mathcal{O}(n\log n)$} (FfFg);
        \draw [->] (FfFg) -- node [above] {$\mathcal{O}(n)$} (FfoFg);
        \draw [->] (FfoFg) -- node [left] {$\mathcal{O}(n \log n)$} (f*g);
    \end{tikzpicture}
    \caption{Computing a 1D-convolution with discrete/fast Fourier transformations}
    \label{fig:DFTvsFFT}
\end{figure}
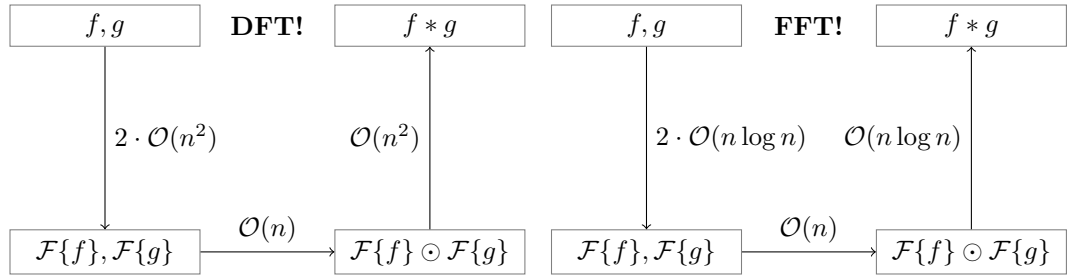

Finally, turning back to the general case $n = (n_1,\dots,n_d)$ we can compute the multi-dimensional \ac{FFT} of $f$ by computing $1$-dimensional \acp{FFT}.
In more detail, we may divide the whole computation into $N/n_j$ $1$-dimensional transformations of size $n_j$ for each coordinate direction $j$. This yields the usual running time of
\[\sum_{j=1}^d \frac{N}{n_j}\mathcal{O}(n_j\log n_j) = \mathcal{O}(N\sum_{j=1}^d\log n_j) = \mathcal{O}(N\log N).\]

\section{Computational Results for E3, E4, and BIG}
\label{apx:results}

\begin{table}[h!]
\caption{Computational results for the classes of family E3}
\label{tab:E3}
\centering
\scriptsize
\begin{tabular}{rrrrrrcc}
\toprule
family & $m$ & $n$ & $U$ & better & equal & avg\_quot & avg\_time\\
\midrule
E3 & $10$ & $31$ & $[100,200]$ & $0$ & $0$ & $1.07$ & $182$\\
E3 & $10$ & $31$ & $[1,100]$ & $0$ & $0$ & $1.08$ & $82$\\
E3 & $10$ & $32$ & $[100,200]$ & $0$ & $2$ & $1.07$ & $179$\\
E3 & $10$ & $32$ & $[1,100]$ & $0$ & $0$ & $1.09$ & $81$\\
E3 & $10$ & $41$ & $[100,200]$ & $0$ & $0$ & $1.07$ & $214$\\
E3 & $10$ & $41$ & $[1,100]$ & $0$ & $0$ & $1.09$ & $80$\\
E3 & $10$ & $42$ & $[100,200]$ & $0$ & $0$ & $1.07$ & $220$\\
E3 & $10$ & $42$ & $[1,100]$ & $0$ & $0$ & $1.1$ & $82$\\
E3 & $10$ & $51$ & $[100,200]$ & $0$ & $0$ & $1.08$ & $225$\\
E3 & $10$ & $51$ & $[1,100]$ & $0$ & $0$ & $1.1$ & $80$\\
E3 & $10$ & $52$ & $[100,200]$ & $0$ & $0$ & $1.06$ & $224$\\
E3 & $10$ & $52$ & $[1,100]$ & $0$ & $0$ & $1.09$ & $75$\\
E3 & $3$ & $10$ & $[100,200]$ & $4$ & $15$ & $1.03$ & $114$\\
E3 & $3$ & $10$ & $[1,100]$ & $3$ & $21$ & $1.03$ & $67$\\
E3 & $3$ & $11$ & $[100,200]$ & $21$ & $2$ & $1.02$ & $150$\\
E3 & $3$ & $11$ & $[1,100]$ & $4$ & $20$ & $1.03$ & $61$\\
E3 & $3$ & $13$ & $[100,200]$ & $20$ & $7$ & $1.03$ & $133$\\
E3 & $3$ & $13$ & $[1,100]$ & $4$ & $15$ & $1.03$ & $60$\\
E3 & $3$ & $14$ & $[100,200]$ & $20$ & $1$ & $1.03$ & $146$\\
E3 & $3$ & $14$ & $[1,100]$ & $6$ & $7$ & $1.03$ & $62$\\
E3 & $3$ & $16$ & $[100,200]$ & $30$ & $2$ & $1.02$ & $166$\\
E3 & $3$ & $16$ & $[1,100]$ & $9$ & $13$ & $1.03$ & $67$\\
E3 & $3$ & $17$ & $[100,200]$ & $7$ & $1$ & $1.03$ & $146$\\
E3 & $3$ & $17$ & $[1,100]$ & $5$ & $11$ & $1.03$ & $59$\\
E3 & $5$ & $16$ & $[100,200]$ & $1$ & $1$ & $1.05$ & $135$\\
E3 & $5$ & $16$ & $[1,100]$ & $0$ & $3$ & $1.06$ & $66$\\
E3 & $5$ & $17$ & $[100,200]$ & $8$ & $2$ & $1.03$ & $149$\\
E3 & $5$ & $17$ & $[1,100]$ & $2$ & $4$ & $1.05$ & $68$\\
E3 & $5$ & $21$ & $[100,200]$ & $19$ & $3$ & $1.02$ & $184$\\
E3 & $5$ & $21$ & $[1,100]$ & $1$ & $0$ & $1.06$ & $65$\\
E3 & $5$ & $22$ & $[100,200]$ & $4$ & $2$ & $1.04$ & $178$\\
E3 & $5$ & $22$ & $[1,100]$ & $1$ & $1$ & $1.06$ & $63$\\
E3 & $5$ & $26$ & $[100,200]$ & $11$ & $1$ & $1.04$ & $185$\\
E3 & $5$ & $26$ & $[1,100]$ & $2$ & $2$ & $1.06$ & $67$\\
E3 & $5$ & $27$ & $[100,200]$ & $1$ & $0$ & $1.05$ & $166$\\
E3 & $5$ & $27$ & $[1,100]$ & $2$ & $1$ & $1.05$ & $67$\\
E3 & $8$ & $25$ & $[100,200]$ & $1$ & $0$ & $1.06$ & $176$\\
E3 & $8$ & $25$ & $[1,100]$ & $0$ & $0$ & $1.07$ & $74$\\
E3 & $8$ & $26$ & $[100,200]$ & $0$ & $0$ & $1.06$ & $177$\\
E3 & $8$ & $26$ & $[1,100]$ & $0$ & $0$ & $1.07$ & $81$\\
E3 & $8$ & $33$ & $[100,200]$ & $0$ & $1$ & $1.06$ & $202$\\
E3 & $8$ & $33$ & $[1,100]$ & $0$ & $0$ & $1.09$ & $72$\\
E3 & $8$ & $34$ & $[100,200]$ & $1$ & $0$ & $1.06$ & $195$\\
E3 & $8$ & $34$ & $[1,100]$ & $0$ & $0$ & $1.08$ & $78$\\
E3 & $8$ & $41$ & $[100,200]$ & $0$ & $0$ & $1.06$ & $207$\\
E3 & $8$ & $41$ & $[1,100]$ & $0$ & $0$ & $1.08$ & $76$\\
E3 & $8$ & $42$ & $[100,200]$ & $0$ & $0$ & $1.06$ & $211$\\
E3 & $8$ & $42$ & $[1,100]$ & $0$ & $0$ & $1.08$ & $77$\\
\bottomrule
\end{tabular}
\end{table}

\begin{table}[h!]
\caption{Computational results for the classes of family E4}
\label{tab:E4}
\centering
\scriptsize
\begin{tabular}{rrrrrrcc}
\toprule
family & $m$ & $n$ & $U$ & better & equal & avg\_quot & avg\_time\\
\midrule
E4 & $2$ & $10$ & $[1,20]$ & $4$ & $61$ & $1.01$ & $39$\\
E4 & $2$ & $10$ & $[1,100]$ & $14$ & $21$ & $1.01$ & $48$\\
E4 & $2$ & $10$ & $[20,50]$ & $10$ & $12$ & $1.02$ & $46$\\
E4 & $2$ & $10$ & $[50,100]$ & $14$ & $9$ & $1.03$ & $59$\\
E4 & $2$ & $10$ & $[100,200]$ & $15$ & $3$ & $1.02$ & $71$\\
E4 & $2$ & $10$ & $[100,800]$ & $14$ & $8$ & $1.01$ & $97$\\
E4 & $3$ & $9$ & $[1,20]$ & $0$ & $43$ & $1.03$ & $37$\\
E4 & $3$ & $9$ & $[1,100]$ & $6$ & $29$ & $1.03$ & $57$\\
E4 & $3$ & $9$ & $[20,50]$ & $2$ & $9$ & $1.05$ & $61$\\
E4 & $3$ & $9$ & $[50,100]$ & $2$ & $0$ & $1.05$ & $75$\\
E4 & $3$ & $9$ & $[100,200]$ & $2$ & $0$ & $1.05$ & $95$\\
E4 & $3$ & $9$ & $[100,800]$ & $5$ & $14$ & $1.03$ & $132$\\
\bottomrule
\end{tabular}
\end{table}

\begin{table}[h!]
\caption{Computational results for the classes of the additional family BIG}
\label{tab:BIG}
\centering
\scriptsize
\begin{tabular}{rrrrrrcc}
\toprule
family & $m$ & $n$ & $U$ & better & equal & avg\_quot & avg\_time\\
\midrule
BIG & $25$ & $100$ & $[1,1000]$ & $0$ & $0$ & $1.13$ & $171$\\
BIG & $50$ & $200$ & $[1,1000]$ & $0$ & $0$ & $1.14$ & $172$\\
BIG & $75$ & $300$ & $[1,1000]$ & $0$ & $0$ & $1.14$ & $179$\\
BIG & $100$ & $400$ & $[1,1000]$ & $0$ & $0$ & $1.15$ & $180$\\
\bottomrule
\end{tabular}
\end{table}

\end{document}